\documentclass[10pt,journal,cspaper,compsoc]{IEEEtran}

\title{CSD Homomorphisms Between Phylogenetic Networks
}

\author{Stephen J. Willson\\
		Department of Mathematics\\
		Iowa State University\\
		Ames, IA 50011 USA\\
		swillson@iastate.edu}

\usepackage{amsmath}
\usepackage{amsthm}
\usepackage{amssymb}

\usepackage[dvips]{graphicx}

\usepackage{verbatim} 
\newtheorem{lem}{Lemma}[section]
\newtheorem{thm}[lem]{Theorem}
\newtheorem{cor}[lem]{Corollary}

\begin{document}

\IEEEcompsoctitleabstractindextext{%
\begin{abstract} Since Darwin, species trees have been used as a simplified description of the relationships which summarize the complicated network $N$ of reality.  Recent evidence of hybridization and lateral gene transfer, however, suggest that there are situations where trees are inadequate.  Consequently it is important to determine properties that characterize networks closely related to $N$ and possibly more complicated than trees but lacking the full complexity of $N$.  

A connected surjective digraph map (CSD) is a map $f$ from one network $N$ to another network $M$ such that every arc is either collapsed to a single vertex or is taken to an arc, such that $f$ is surjective, and such that the inverse image of a vertex is always connected.  CSD maps are shown to behave well under composition.  It is proved that if there is a CSD map from $N$ to $M$, then there is a way to lift an undirected version of $M$ into $N$, often with added resolution.  A CSD map from $N$ to $M$ puts strong constraints on $N$.  

In general, it may be useful to study classes of networks such that, for any $N$, there exists a CSD map from $N$ to some standard member of that class.  
\end{abstract}
}

\maketitle

\IEEEdisplaynotcompsoctitleabstractindextext



\section{ Introduction}  

This paper is a modified version of \cite{wil12}.  This version slightly generalizes the notion there of a CSD map but obtains essentially the same results with essentially the same proofs.  

Since Darwin, phylogenetic trees have been used to display relationships among species, and they have become a standard tool in phylogeny.  More recently, in order to deal with the possibilities of such events as hybridization and lateral gene transfer, more general phylogenetic networks have become of interest  \cite{mor04}, \cite{nak04}, \cite{doo07}, \cite{dag08},  \cite{clr08}, \cite{mor09}.    Different researchers have found it useful to make a broad range of assumptions about the networks in order to be able to obtain various results.  

The underlying reality for, say, successive sexually reproducing populations of diploid organisms, is a complicated network $N$ of parents and children of individual organisms---a full genealogy reaching back to ancestors in the remote past.  Trying to reconstruct such a reality from extant taxa is a hopeless goal.  Instead, we have often relied on a species tree $T$ as a phylogeny at a more abstract level.  In principle, the underlying complicated network $N$ has been usefully transformed into the much simpler species tree $T$.  

This paper explores relationships between $N$ and other related networks $M$, potentially much simpler than $N$, but perhaps more complicated than trees.  Other researchers have looked at similar problems.  General frameworks for networks are discussed in \cite{ban92}, \cite{bar04}, \cite{mor04}, and \cite{nak04}.  Typically these frameworks model phylogenies by acyclic rooted directed graphs. 
Wang {\it et al.} \cite{wan01} and Gusfield {\it et al.} \cite{gel04a} study ``galled trees'' in which all recombination events are associated with node-disjoint recombination cycles.  Van Iersel and others generalized galled trees to ``level-$k$'' networks \cite{ier09}.  Baroni, Semple, and Steel \cite{bar04} introduced the idea of a ``regular" network, which coincides with its cover digraph.  Cardona {\it et al.} \cite{crv09} discussed ``tree-child" networks, in which every vertex not a leaf has a child that is not a reticulation vertex.  Dress {\it et al.}  \cite{dre10} consider alternative ways to derive trees, or, more generally, hierarchies from a network.  

Let $N$ and $M$ be phylogenetic $X$-networks.   Such networks are rooted directed graphs with specified leaf set $X$.  (Further details are given in section 2).  The basic tool studied in this paper is that of a connected surjective digraph (CSD) map $f: N \to M$.  A formal definition is in section 3, but, roughly, such a map $f$ is a map on the vertex sets,
 $f: V(N) \to V(M)$, satisfying \\
(1) $f$ is onto; \\
(2) whenever $(u,v)$ is an arc of $N$, then either $(f(u), f(v))$ is an arc of $M$, or else $f(u) = f(v)$, and every arc of $M$ arises in this manner;\\
(3) for each vertex $v'$ of $M$, $f^{-1}(v')$ consists of the vertices of a connected subgraph of $N$.

CSD maps are special cases of graph homomorphisms, which have been the subject of recent investigations, including a recent book \cite{hel04} by Hell and Ne\v set\v ril.  A review of graph homomorphisms, especially with applications to colorings, is in Hahn and Tardif \cite{hah97}.  These studies do not include studies of homomorphisms with property (3).  Work by Daneshgar {\it et al.} \cite{dan08} concerns ``connected graph homomorphisms'' but with a very different notion of connectedness, requiring that the inverse image of an edge be empty or connected.

Figure 1 shows a network $N$ and a network $N'$ which happens to be a tree.  There is a CSD map $f : N \to N'$.  Each vertex $v$ in $N$ is labelled by the name of the vertex $f(v)$ in $N'$.  The set of leaves, corresponding to extant taxa, is $X = \{1,2,3,4\}$.  In this particular case, the tree $N'$ is a plausible candidate for the ``species tree'' corresponding to $N$.  

The networks $M$ for which there is a CSD map from $N$ to $M$ are seen in section 3 to arise as certain quotient structures of $N$ in a natural way.

\begin{figure}[!htb]  
\begin{center}

\begin{picture}(200,390) 
\thicklines
\put(100,380){\vector(1,-1){30}}
\put(100,380){\vector(-1,-1){30}}
\put(70,350){\vector(-1,-1){30}}
\put(70,350){\vector(1,-1){30}}
\put(40,320){\vector(-1,-1){30}}
\put(10,290){\vector(0,-1){30}}
\put(10,260){\vector(0,-1){30}}
\put(10,230){\vector(0,-1){30}}
\put(10,200){\vector(0,-1){30}}
\put(10,170){\vector(1,-1){30}}
\put(40,260){\vector(1,-1){30}}
\put(70,230){\vector(1,-1){30}}
\put(100,200){\vector(0,-1){30}}
\put(100,170){\vector(0,-1){30}}
\put(100,320){\vector(1,-1){30}}
\put(100,320){\vector(-1,-1){30}}
\put(70,290){\vector(-1,-1){30}}
\put(130,290){\vector(1,-1){30}}
\put(160,260){\vector(-1,-1){30}}
\put(130,230){\vector(1,-1){30}}

\put(101,380){\vector(1,-1){30}}   
\put(101,380){\vector(-1,-1){30}}
\put(71,350){\vector(-1,-1){30}}
\put(71,350){\vector(1,-1){30}}
\put(41,320){\vector(-1,-1){30}}
\put(11,290){\vector(0,-1){30}}
\put(11,260){\vector(0,-1){30}}
\put(11,230){\vector(0,-1){30}}
\put(11,200){\vector(0,-1){30}}
\put(11,170){\vector(1,-1){30}}
\put(41,260){\vector(1,-1){30}}
\put(71,230){\vector(1,-1){30}}
\put(101,200){\vector(0,-1){30}}
\put(101,170){\vector(0,-1){30}}
\put(101,320){\vector(1,-1){30}}
\put(101,320){\vector(-1,-1){30}}
\put(71,290){\vector(-1,-1){30}}
\put(131,290){\vector(1,-1){30}}
\put(161,260){\vector(-1,-1){30}}
\put(131,230){\vector(1,-1){30}}

\thinlines
\put(10,260){\vector(1,-1){30}}
\put(70,230){\vector(0,-1){30}}
\put(40,320){\vector(1,-1){30}}
\put(70,290){\vector(1,-1){30}}
\put(100,260){\vector(1,-1){30}}
\put(10,290){\vector(1,-1){30}}
\put(40,260){\vector(-1,-1){30}}
\put(100,260){\vector(-2,-1){60}}
\put(130,290){\vector(-1,-1){30}}
\put(40,230){\vector(-1,-1){30}}
\put(40,230){\vector(0,-1){30}}
\put(40,230){\vector(1,-1){30}}
\put(70,200){\vector(-1,-1){30}}
\put(70,200){\vector(1,-1){30}}
\put(100,200){\vector(1,-1){30}}
\put(130,170){\vector(-1,-1){30}}
\put(40,200){\vector(-1,-1){30}}
\put(10,200){\vector(1,-1){30}}
\put(40,170){\vector(0,-1){30}}

\put(110, 380){1234}
\put(40,380){$N$}
\put(135,350){1}
\put(50,350){234}
\put(20,320){234}
\put(110,320){234}
\put(-4,290){34}
\put(80, 290){234}
\put(135,290){234}
\put(-4,260){34}
\put(50,260){34}
\put(110,260){234}
\put(165,260){$2$}
\put(140,230){$2$}
\put(165,200){2}
\put(0,230){$4$}
\put(40,240){34}
\put(75,230){34}
\put(0,200){$4$}
\put(45,200){$4$}
\put(75,200){34}
\put(105,200){$3$}
\put(0,170){$4$}
\put(45,165){$4$}
\put(105,170){$3$}
\put(135,170){$3$}
\put(45,140){$4$}
\put(105,135){$3$}

\put(100,90){\vector(-1,-1){30}}
\put(100,90){\vector(1,-1){30}}
\put(70,60){\vector(-1,-1){30}}
\put(70,60){\vector(1,-1){30}}
\put(40,30){\vector(-1,-1){30}}
\put(40,30){\vector(1,-1){30}}

\put(40,65){$N'$}
\put(100,95){1234}
\put(110,60){1}
\put(80,60){234}
\put(50,20){34}
\put(105,20){$2$}
\put(0,0){$4$}
\put(74,0){$3$}

\end{picture}

\caption{   Two $X$-networks $N$ and $N'$, in which $N'$ happens to be a tree with topology (1,(2,(3,4))).  
There is a CSD map $f$ from $N$ to $N'$, given by the labelling of vertices in $N$. In fact, section 5 shows $N' = ClDis(N)$.   A certain tree $U$ with topology (1,(4,(2,3)))   displayed by $N$ as a subgraph is shown in bold. There is, however, no CSD map from $N$ to $U$.  }
\end{center}
\end{figure}
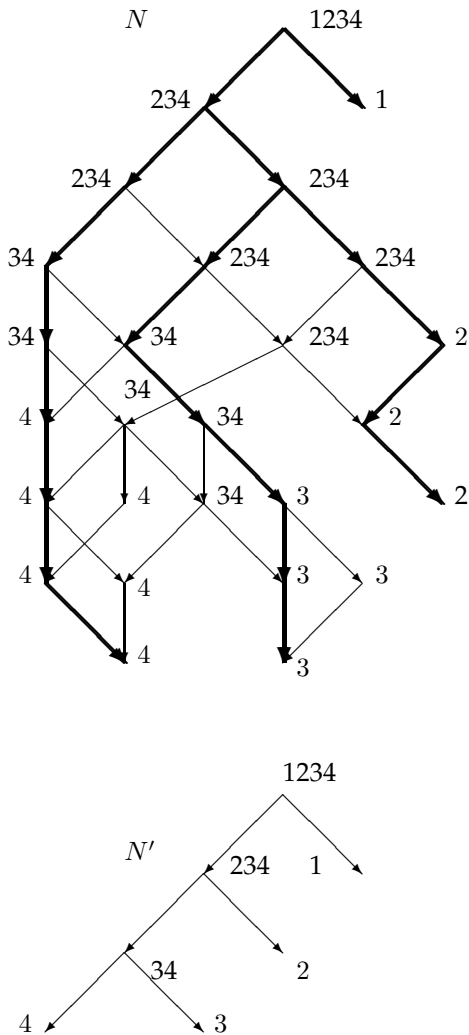

The condition that there be a CSD map $f: N \to N'$  is very different from the condition that $N'$ be displayed by $N$; {\it i.e.} that $N$ contain a directed $X$-subgraph homeomorphic with $N'$.  If $N$ is the network in Figure 1,  there is a CSD map $f$ from $N$ to the tree $N'$ with topology (1,(2,(3,4))).  It is true that the tree $(1,(2,(3,4)))$ is displayed in $N$.   But another tree $U$ with topology (1,(4,(2,3))) shown in bold in Figure 1 is also displayed in $N$, yet there is no CSD map from $N$ to $U$.   If we restrict the map $f$ to the tree $U$ to yield the map $f|U$ from $U$ to $N'$, then $f|U$ remains a surjective digraph map from $U$ onto $N'$, but it is not connected since the preimage of vertex 34 is no longer connected.  

The  essential condition for a CSD map $f:N\to M$ is (3), that for each vertex $v$ of M the points of $N$ mapping to $v$ induce a connected subgraph of $N$.  In retrospect, this condition appears natural:  The essential topological property of a single point is that it is connected, {\it i.e.,} all in one piece.  The essential biological property of a single population is that it is connected, since an organism arises only from another organism.  In order to find natural relationships among networks $N$ and $M$, we assume here that the points of $N$ corresponding to a single vertex in $M$ should therefore also be connected.  

In this paper it is proved (Theorem 4.1) that whenever $f: N \to N'$ is a CSD map, then $N'$ can be ``lifted'' into $N$ possibly in many ways, each called a \emph{wired lift} in this paper.  Any wired lift is an undirected subgraph of $N$ resembling $N'$ but possibly containing more resolution.  Thus some aspects of $N'$ also are exhibited in $N$.  The fact that each $f^{-1}(v')$ is connected is essential to this possibility.  More generally, if $f: N \to N'$ is a CSD map, then $N'$ places strong constraints on the structure of $N$.  In contrast, it is shown (Theorem 4.5) that without the connectedness property, the constraints on $N$ would be minimal.  

For example, it is common to study trees or more generally networks in terms of the quartets or rooted triples that they exhibit \cite{sem03} p. 116.     Suppose $f: N \to N'$ is a CSD map and $N'$ displays a particular quartet $ab|cd$ in the sense that the undirected graph of $N'$ contains a subgraph homeomorphic with the given quartet.  Using the wired lift, Corollary 4.3 shows that the undirected version of $N$ also contains a subgraph homeomorphic with the given quartet.  Hence quartet information in $N'$ is preserved in $N$, again placing strong constraints on the structure of $N$.

Suppose $X$ denotes the leaf set of the networks, corresponding to the set of extant species on which measurements may be made.  Following \cite{bar04} define the \emph{cluster} of a vertex $v$ in the network $N$, denoted $cl(v,N)$, to be the set of members of $X$ which are descendents of $v$.  A network $N$ is called \emph{successively cluster-distinct} if, whenever $(u,v)$ is an arc of $N$, then $cl(u,N) \neq cl(v,N)$.  

In section 5, given $N$ we show how to construct a well-defined network $ClDis(N)$ which is successively cluster-distinct and such that there is a CSD map $f: N \to ClDis(N)$.  For example, if $N$ is the network in Figure 1, then  $N' = ClDis(N)$.  The network $ClDis(N)$ potentially is vastly simpler than $N$, although it need not be a tree in general.  The wired lift of $ClDis(N)$ into $N$ shows that in some sense $ClDis(N)$ can act as a ``skeleton'' of $N$.  It is shown (Corollary 5.5) that $ClDis(N)$ has a ``universal'' property uniquely identifying it among all  cluster-distinct networks related to $N$.  This raises possible interest in the study of successively cluster-distinct networks as a tool for studying general phylogenetic networks.  

One theme of this paper is the transformation of networks.  Common standard operations to transform  graphs include:\\
(1) contraction of an arc or edge to a point; and  \\
(2) deletion of an arc or edge.  \\
Theorem 3.6 puts CSD maps in this broader context.  Suppose $f: N \to N'$ is a CSD map.  Then there is a way to obtain $N'$ by starting with $N$ and recursively contracting arcs to points without ever explicitly deleting any arcs.

Section 6 discusses some implications of these results.

\section {Fundamental Concepts} 

A \emph{directed graph} or \emph{digraph} $N=(V,A)$ consists of a finite set $V$ of \emph{vertices} and a finite set $A$ of \emph{arcs}, each consisting of an ordered pair $(u,v)$ where $u \in V$, $v \in V$, $u\neq v$.  Sometimes we write $V(N)$ for $V$ or $A(N)$ for $A$.   We interpret $(u,v)$ as an arrow from $u$ to $v$ and say that the arc \emph{starts} at $u$ and \emph{ends} at $v$.  There are no multiple arcs and no loops.   If $(u,v)\in A$, say that $u$ is a \emph{parent} of $v$ and $v$ is a \emph{child} of $u$.  A \emph{directed path} is a sequence $u_0, u_1, \cdots, u_k$ of vertices such that for $i = 1, \cdots, k$, $(u_{i-1}, u_i) \in A$.  The path is \emph{trivial} if $k = 0$.  Write $u \leq  v$ if there is a directed path starting at $u$ and ending at $v$.  The digraph is \emph{acyclic} if there is no nontrivial directed path starting and ending at the same point.  If the digraph is acyclic, it is easy to see that $\leq$  is a partial order on $V$.  

The \emph{indegree} of vertex $u$ is the number of $v \in V$ such that $(v,u) \in A$.  The \emph{outdegree} of $u$ is the number of $v \in V$ such that $(u,v) \in A$.  
A \emph{leaf} is a vertex of outdegree 0. A \emph{normal} (or \emph{tree-child}) vertex is a vertex of indegree 1.  A \emph{hybrid} vertex (or \emph{recombination vertex} or \emph{reticulation node}) is a vertex of indegree at least 2.

The digraph $(V,A)$ is \emph{rooted} if there exists  a unique vertex $r \in V$ with indegree 0 such that for all $v \in V$, $r\leq v$.  The vertex $r$ is called the \emph{root}.

Let $X$ denote a finite set.  Typically in phylogeny, $X$ is a collection of species.  An \emph{$X$-network} $(V,A,r,X)$ is a 
digraph $G=(V,A)$ with root $r$ such that \\
(1) there is a one-to-one map 
$\phi: X \to V$ 
such that the image of $\phi$ is the set of all leaves of $G$, and \\
(2) for every $v \in V$ there is a leaf $u$ and a directed path from $v$ to $u$. \\
Thus the set of leaves of $G$ may be identified with the set $X$ and every vertex is ancestral to a leaf.  

In biology most $X$-networks are acyclic.  The set $X$ provides a context for $G$, giving a hypothesized relationship among the members of $X$.  For convenience, we will write $x$ for the leaf $\phi(x)$.  

An \emph{$X$-tree} is an $X$-network such that the underlying digraph is a rooted tree.

Let $N = (V,A,r,X)$ and $N' = (V',A',r',X)$ be $X$-networks.  An \emph{X-isomorphism} $\psi: N \to N'$ is a map $\psi: V \to V'$ such that\\
(1) $\psi: V \to V'$ is one-to-one and onto, \\
(2) $\psi(r) = r'$,\\
(3) for each $x \in X$, $\psi(x) = x$,\\
(4) $(\psi(u), \psi(v))$ is an arc of $N'$ iff $(u,v)$ is an arc of $N$.\\
We say $N$ and $N'$ are \emph{isomorphic} if there is an $X$-isomorphism $\psi: N \to N'$.

A \emph{graph} (or, for emphasis, an \emph{undirected graph}) $(V,E)$ consists of a finite set $V$ of \emph{vertices} and a finite set $E$ of \emph{edges}, each consisting of a subset $\{v_1,v_2\}$ where $v_1$ and $v_2$ are  two distinct members of $V$.  Thus an edge has no direction, while an arc has a direction.  If $u \in V$, then the \emph{total degree} of $u$ is the number of edges in $E$ containing $u$.  If $G=(V,E)$ is a graph and $W$ is a subset of $V$, the \emph{induced subgraph} $G[W]$ is the graph $(W,E[W])$ where the edge set $E[W]$ is the collection of all $\{v_1,v_2\}$ in $E$ such that $v_1 \in W$ and $v_2 \in W$.  Thus $G[W]$ contains all edges both of whose endpoints are in $W$. 

If $G=(V,E)$ is a graph and $\{v_1,v_2\}$ is an edge, then a new graph $G'=(V',E')$ may be obtained by adding a new vertex $v_3 \notin V$, removing $\{v_1,v_2\}$ and adding two new edges $\{v_1,v_3\}$ and $\{v_2,v_3\}$.  Thus the new vertex $v_3$ has total degree 2 in $G'$.  We say that $G$ is obtained from $G'$ by \emph{suppressing} the vertex $v_3$ of total degree 2 and $G'$ is obtained from $G$ by \emph{inserting} the vertex $v_3$ of total degree 2.  We say that $G$ and $G''$ are \emph{homeomorphic} if there is a sequence $G = G_0, G_1, \cdots, G_k$ of graphs such that for $i = 1, \cdots, k$,  $G_i$ is obtained from $G_{i-1}$ either by inserting a vertex of total degree 2 or by suppressing a vertex of total degree 2.  

A graph $G=(V,E)$ is \emph{connected} if, given any two distinct $v$ and $w$ in $V$ there exists a sequence
$v = v_0, v_1, v_2, \cdots, v_k = w$ of vertices  such that for $i = 0, \cdots, k-1$, $\{v_i, v_{i+1}\}\in E$.  
A subset $W$ of $V$ is \emph{connected} if the induced subgraph $G[W]$ is connected.

Given a digraph $G=(V,A)$ define $Und(G) = (V,E)$ where $E = \{\{u,v\}:$ there is an arc $(u,v) \in A\}$.  Then $Und(G)$ is an undirected graph with the same vertex set as $G$ and with edges obtained by ignoring the directions of arcs.  A subset $W$ of $V$ is \emph{connected} if the induced subgraph $Und(G)[W]$ is connected.   Thus a connected subset of $G$ is defined ignoring the directions of arcs.

\section{ New Connected Surjective Digraph Maps} 

Let $N=(V,A,r,X)$ and $N'=(V',A',r',X)$ be  $X$-networks whose leaf sets are identified with the same set $X$.  A \emph{ new X-digraph map}   $f: N\to N'$
is a map $f: V \to V'$ such that\\
(a) $f(r) = r'$,\\
(b) for all $x \in X$, $f(x) = x$, and \\
(c) if $(u,v)$ is an arc of $N$, then there is a directed path in $N'$ from $f(u)$ to $f(v)$.
A new X-digraph map f is a \emph{strong} or an \emph{old X-digraph map} provided 
(a) $f(r) = r'$,\\
(b) for all $x \in X$, $f(x) = x$, and \\
(c) if $(u,v)$ is an arc of $N$, then either $f(u) = f(v)$ or else $(f(u), f(v))$ is an arc of $N'$.\\
A strong X-digraph map is the same as the notion of X-digraph map in \cite{wil12}. Note that a strong map is clearly also a new map, but the converse is false.   

This allows the possibility that $f(u) = f(v)$ or $(f(u), f(v))$ is an arc of $N'$, as well as other possibilities.  

If $v' \in V'$, the set $f^{-1}(v')$ consists of all vertices $u \in V$ such that $f(u) = v'$.  In particular, if $(u,v)$ is an arc of $N$ and there is a directed path in $N'$ from $f(u)$ to $f(v)$ which passes through $v'$ but neither $f(u) = v'$ nor $f(v) = v'$, then neither $u$ nor $v$ is in $f^{-1}(v')$, nor does the directed path in $N'$ from $f(u)$ to $f(v)$ contribute any point to $f^{-1}(v')$.  

Call $f$ \emph{new connected} if for each $v' \in V'$, $f^{-1}(v')\subseteq V$ is a connected subset of $N$, {\it i.e.}, if the induced subgraph $Und(N)[f^{-1}(v')]$ is  connected.  Call $f$ \emph{surjective} if for each $v' \in V'$, $f^{-1}(v')$ is a nonempty subset of $V$ and for each arc $(a,b)$ of $N'$ there exist vertices $u$ and $v$ of $N$ such that $(u,v)$ is an arc of $N$, $f(u) = a$, and $f(v) = b$.  
The \emph{kernel} of $f$ is the partition $\{f^{-1}(v') :v' \in V'\}$ of $V$.  

We are interested primarily in new $X$-digraph maps that are both new connected and surjective.  They will be called \emph{new connected surjective digraph maps}  or \emph{NCSD maps}.  Many of their properties are analogous to properties of homomorphisms \cite{hel04} but properties involving the leaf set $X$ and connectivity require special attention. 

Let $N=(V,A,r,X)$ be an $X$-network, where $\phi:X \to V$ gives the identification.  Suppose $\sim$ is an equivalence relation on $V$.   Let  $[v]$ denote the equivalence class of $v \in V$.  The \emph{partition} of $V$ into equivalence classes is  $\mathcal{P} = \{[v]: v \in V\}$. 
The equivalence relation $\sim$ and the partition $\mathcal{P}$ are  called \emph{leaf-preserving} provided that  no two distinct leaves are equivalent and, in addition, for every $x \in X$ whenever $u \in [x]$ and $(u,v)$ is an arc, then $v \in [x]$. 

Let $N=(V,A,r,X)$ be an $X$-network.  Suppose $\sim$ is an equivalence relation on $V$ with partition $\mathcal{P}$.  Define the \emph{quotient digraph} $N'$   by $N' = (V',A',r',X)$ where \\
(i) $V' = \mathcal{P}$ is the set of equivalence classes $[v]$.\\
(ii) $r' = [r]$.\\
(iii) The member $x\in X$ corresponds to $[x]$; {\it i.e.}, the identification is given by $\phi': X \to V'$ by $\phi'(x) = [\phi(x)]$.\\
(iv) Let $[u]$ and $[v]$ be two equivalence classes.  There is an arc $([u], [v]) \in A'$ iff
$[u] \neq [v]$  and there exists $u' \in [u]$ and $v' \in [v]$ such that $(u',v')\in A$.\\
Alternative notations for $N'$ will be  $N/\sim$ or $N/\mathcal{P}$.  

\begin{thm}
Let $N = (V,A,r,X)$ be an  $X$-network.  Suppose $\sim$ is a leaf-preserving equivalence relation on $V$.  Let $N' = N/\sim  \:= (V',A',r',X)$ be the quotient digraph. Then \\
(1) $N'$ is an $X$-network.\\
(2) The natural map $\phi: N \to N'$ 
given by $\phi(u) = [u]$
is a surjective strong $X$-digraph map with kernel the set of equivalence classes under $\sim$.\\
(3)  If each equivalence class $[u]$ is connected in $N$, then $\phi$ is new connected.
\end{thm}

\begin{proof}
(1) It is immediate that $(V',A')$ is a directed graph with no loops and no multiple arcs.  If $u_0, u_1, \cdots, u_k$ is a directed path in $N$ (so for $i = 0, \cdots, k-1$, $(u_i,u_{i+1}) \in A$), then $[u_0], [u_1], \cdots, [u_k]$ is a sequence of vertices in $N'$ and for each $i = 0, \cdots, k-1$,  there is a directed path from $[u_i]$ to $u_{i+1}$ in $A'$.   It follows that $r'$ is a root of $N'$. 

Suppose $x\in X$; we show that $[x]$ is a leaf of $N'$.  Suppose there is an arc $([x] ,[y])$.  Then there exist $a\in [x]$ and $b\in [y]$ such that $(a,b) \in A$.  Since $\sim$ is leaf-preserving, $b \in [x]$ so $[x]=[y]$, contradicting that there are no loops in $(V',A')$.

Conversely, suppose that $[u]$ is a leaf of $N'$; I claim that there exists $x\in X$ such that $[u] = [x]$.  If not, then no vertex of $N$ in $[u]$ is a leaf, since $X$ is identified with the set of leaves.  Since $N$ is an $X$-network, we may choose a directed path in $N$ starting at $u$ to some leaf $x$.  Since $x$ is a leaf, $x \notin [u]$, so $N'$ has an arc from $[u]$ to some other vertex, contradicting that $[u]$ is a leaf.   

Finally, given a vertex $[u] \in V'$, note that there is a leaf $x \in X$ such that $N$ contains a directed path from $u$ to $x$; it follows that in $N'$ there is a directed path from $[u]$ to $[x]$.

(2) We check the conditions (a), (b), and (c) for being a strong $X$-digraph map. Condition (a) is immediate.  For (b), note that if $x \in X$, then $\phi(x) =[x]$.  To see (c), suppose $(u,v)$ is an arc of $N$.  Then either $[u] = [v]$ or else $([u], [v])$ is an arc of $N'$.  To see surjectivity, it is immediate that $\phi^{-1}([u]) = [u]$ is nonempty.  Given an arc $([u],[v])$ of $N'$ there exist $u' \in [u]$ and $v' \in [v]$ such that $(u',v') \in A$, but then $\phi(u') = [u]$ and $\phi(v') = [v]$.  To see surjectivity, it is immediate that $\phi^{-1}([u]) = [u]$ is nonempty.  Given an arc $([u],[v])$ of $N'$ there exist $u' \in [u]$ and $v' \in [v]$ such that $(u',v') \in A$, but then $\phi(u') = [u]$ and $\phi(v') = [v]$. 

(3) follows since $\phi^{-1}([u]) = [u]$.
\end{proof}

If $N$ is acyclic, it need not follow that $N'$ is also acyclic.  

The following converse shows that the image of a surjective digraph map is essentially the same as the natural quotient digraph.  

\begin{thm}
Let $N = (V,A,r,X)$ and $N' = (V',A',r',X)$ be $X$-networks. 
Suppose $f: N \to N'$ is a surjective strong $X$-digraph map.  Define the relation $\sim$ on $V$  by $u\sim v$ iff $f(u) = f(v)$.   Then $\sim$ is a leaf-preserving equivalence relation and the equivalence classes are $[u] = f^{-1}(f(u))$.  Moreover the quotient digraph $N/\sim$ is isomorphic with $N'$ via the map $\phi: N/\sim \: \to N'$ given by
$\phi ([u]) = f(u)$.
\end{thm}

\begin{proof}
It is immediate that $\sim$ is an equivalence relation and that $\phi$ is one-to-one and onto.  To see that it is leaf-preserving, suppose $x \in X$, $u \in V$ satisfies $u \in f^{-1}(x)$, and $v \in V$ satisfies that $(u,v)$ is an arc. We must show that $v \in f^{-1}(x)$.   But since $f$ is a digraph map, either $f(u) = f(v)$ or $(f(u), f(v))$ is an arc.  In the former case $f(v) = f(u) = x$; in the latter case there is an arc from $f(u) = x$ to $f(v)$, contradicting that $x$ is a leaf in $N'$.

If $([u],[v])$ is an arc of $N/\sim$ then there exist $u' \in [u]$ and $v' \in [v]$ such that $(u',v')$ is an arc of $N$.   Since $f(u') \neq f(v')$ and $f$ is an  $X$-digraph map it follows $(f(u'), f(v'))$ is an arc of $N'$.   Conversely, suppose $(a,b)$ is an arc of $N'$.  Since $f$ is surjective there exist vertices $u$ and $v$ of $N$ such that $(u,v)$ is an arc of $N$, $f(u) = a$, and $f(v) = b$.  Since $a\neq b$ it follows $[u] \neq [v]$, so $([u],[v])$ is an arc of $N/\sim$ which satisfies that $\phi([u]) = a$ and $\phi([v]) = b$.  
\end{proof}

The connectedness of the inverse images of points implies the connectedness of the inverse images of more general connected sets:

\begin{thm} 
 Let $N = (V,A,r,X)$ and $N' = (V',A',r',X)$ be $X$-networks.  Let $f: N \to N'$ be a CSD map.  If $B \subseteq V'$ is connected in $N'$, then $f^{-1}(B)$ is connected in $N$.
\end{thm}

\begin{proof}
Write $B = \{v_1', v_2', \cdots, v_k'\}$.    Then $f^{-1}(B)= \cup [f^{-1}(v_i'): i = 1, \cdots, k]$.  Since $B$ is connected, there exist arcs $(v_{a_i}', v_{b_i}')$ for $i = 1, \cdots, m$ such that  these arcs connect together the members of $B$.   Since $f$ is surjective, for each $i$ there exist vertices $v_{a_i} \in f^{-1}(v_{a_i}')$ and $v_{b_i} \in f^{-1}(v_{b_i}')$ such that $(v_{a_i}, v_{b_i})\in A$.  But now since each set $f^{-1}(v_i')$ is connected, it follows that $f^{-1}(B)$ is connected.
\end{proof}

\begin{thm}
Let $N$,  $N'$, and $N''$ be $X$-networks.  Let $f: N \to N'$ and $g: N' \to N''$ be $X$-digraph maps.\\
(a) The composition $g \circ f : N \to N''$ is an $X$-digraph map.\\
(b)  If $f$ and $g$ are surjective, then $g \circ f$ is surjective.\\
(c) If $f$ and $g$ are connected and surjective, then $g \circ f$ is connected and surjective.
\end{thm}

\begin{proof}

 (a) and (b) are immediate.   For (c), suppose $f$ and $g$ are connected and surjective. From (b), $g \circ f$ is surjective.  For any vertex $v''$ of $N''$, $(g \circ f)^{-1}(v'') = f^{-1}(g^{-1}(v''))$.  Since $g$ is connected, $g^{-1}(v'')$ is connected.  But then by Theorem 3.3 since $f$ is connected, $f^{-1}(g^{-1}(v''))$ is connected.
\end{proof}

It follows that the composition of any number of CSD maps is also a CSD map.  The network which is the image of the last map is thus a quotient digraph of the first network.

We next show that in certain circumstances a CSD map $f$ can be factored as $f = h \circ g$, where $g$ and $h$ are CSD maps. 

Suppose $N=(V,A,r,X)$ is an $X$-network.  
A partition $\mathcal{Q}$ of $V$ is \emph{subordinate} to a partition $\mathcal{P}$ of $V$ provided, for each $A \in \mathcal{Q}$, there exists $B \in \mathcal{P}$ such that $A \subseteq B$.  

\begin{thm}
Let $N=(V,A,r,X)$ and $N'=(V',A',r',X)$ be $X$-networks.  
Let $f: N \to N'$ be a surjective $X$-digraph map with kernel $\mathcal{P} = \{f^{-1}(v): v \in V'\}$.   Suppose $\mathcal{Q}$ is a leaf-preserving partition of $V$ that is subordinate to $\mathcal{P}$.  \\
(1) There exist surjective $X$-digraph maps $g: N \to N/\mathcal{Q}$ and $h: N/\mathcal{Q} \to N'$ such that 
$f = h \circ g$.\\
(2) If in addition $f$ is connected and each member of $\mathcal{Q}$ is connected, then both $h$ and $g$ are connected.\\

\end{thm}

\begin{proof}
(1) Write $[v]_\mathcal{Q}$ for the member of $\mathcal{Q}$ that contains vertex $v$.  Define $g$ by $g(v)  = [v]_\mathcal{Q}$.   
Define $h$ by 
$h([v]_\mathcal{Q}) = [v]_\mathcal{P}$.  To see that $h$ is well-defined we must show that if $[v_1]_\mathcal{Q} = [v_2]_\mathcal{Q}$, then it follows that
$[v_1]_\mathcal{P} = [v_2]_\mathcal{P}$.
But if $[v_1]_\mathcal{Q} = [v_2]_\mathcal{Q}$, then $v_1$ and $v_2$ are in the same member of the partition, whence because $\mathcal{Q}$ is subordinate to $\mathcal{P}$ we have $[v_1]_\mathcal{P} = [v_2]_\mathcal{P}$.  Hence  $h$ is well-defined.   Moreover,  $(h \circ g)(v) = h(g(v))$ $ = h([v]_\mathcal{Q})$ $ = [v]_\mathcal{P} = f(v)$ using Theorem 3.1.

Since $\mathcal{Q}$ is leaf-preserving, the map $g$ is an $X$-digraph map by Theorem 3.1.  There remains to see that $h$ is an $X$-digraph map.   Note if $x\in X$ then $h([x]_\mathcal{Q}) = [x]_\mathcal{P}=f(x)$ is the leaf labelled $x$ in $N'$.  Likewise $h([r]_\mathcal{Q}) = [r]_\mathcal{P}=f(r)$ is the root of $N'$.   The condition on arcs is seen similarly.  Hence $h$ is an $X$-digraph map.

Since $f$ is surjective, for each $v' \in V'$ there exists $v \in V(N)$ such that $f(v) = v'$.  Hence $h([v]_\mathcal{Q}) = v'$ and $g(v) = ([v]_\mathcal{P})$ so $h$ and $g$ are surjective as maps of sets.  If $(u',v')$ is an arc of $N'$, then since $f$ is surjective there exist vertices $u$ and $v$ of $N$ such that $f(u) = u'$, $f(v) = v'$, and $(u,v)$ is an arc of $N$.  Hence $(g(u), g(v))$ is an arc of $N/\mathcal{Q}$ and $h(g(u)) = u'$, $h(g(v)) = v'$ in $N'$, so $h$ is surjective.  Moreover, $g$ is surjective by Theorem 3.1.

For (2) suppose $f$ is connected and each member of $\mathcal{Q}$ is connected.  Each vertex of $N/\mathcal{Q}$ is a subset $B$ of $V$ for $B \in \mathcal{Q}$.  By hypothesis $B$ is connected, so it follows that $g$ is connected.  Next suppose $v \in V'$; since $f$ is surjective, pick $w \in f^{-1}(v)$.  Then $h^{-1}(v)$ is the image in $N/\mathcal{Q}$ of $[w]_\mathcal{P}$.  But $[w]_\mathcal{P}$ is connected since $f$ was connected, so its image in $N/\mathcal{Q}$ is also connected.  Hence $h$ is connected.
\end{proof}

One theme of this paper is the transformation of networks.  Common operations to transform  graphs include:\\
(1) contraction of an arc or edge to a point; and  \\
(2) deletion of an arc or edge.  \\
For example, when a tree has a vertex $v$ of indegree 1 and outdegree 1, it is common to simplify the tree by contracting one of the edges involving $v$ to a point.  When a network $N$ is said to \emph{display} a tree $T$, the meaning is that certain arcs (directed into hybrid vertices) may be deleted  so that the resulting graph is a tree homeomorphic with $T$.  An example of contraction is shown in Figure 2.

In fact, we see below that CSD maps arise from recursive  contractions of arcs  without any deletion of arcs.  

More precisely, given a digraph $N = (V,A)$ with arc $(a,b)$, contraction of the arc to a point means forming a new network $N'$ as follows:\\
(i) Add a new vertex $a'$.\\
(ii) Remove each arc $(u,a)$ with $u \neq b$ and add an arc $(u,a')$.\\
(iii) Remove each arc $(a,u)$ for $u\neq b$ and add an arc $(a',u)$.\\
(iv) Remove each arc $(u,b)$ with $u \neq a$ and add an arc $(u,a')$.\\
(v) Remove each arc $(b,u)$ for $u\neq a$ and add an arc $(a',u)$.\\
(vi) Delete the vertices $a$ and $b$ and the arc $(a,b)$ as well as the arc $(b,a)$ if it existed originally.  

In steps (ii) through (v), note that if the new arc already exists (from some previous step) then we do not add an additional copy.

\begin{figure}[!htb]  
\begin{center}

\begin{picture}(120,110) (0,0)

\put(70,100) {\vector(1,-1){30}}
\put(100,70) {\vector(0,-1){30}}
\put(100,40) {\vector(-1,-1){30}}
\put(70,100) {\vector(0,-1){90}}
\put(68,10) {\vector(0,1){90}}
\put(105,70){$e$}
\put(105,40){$d$}
\put(75,10){$c$}
\put(75,100){$a'$}
\put(115,55){$N'$}

\put(10,100) {\vector(1,-1){30}}
\put(40,70) {\vector(0,-1){30}}
\put(40,40) {\vector(0,-1){30}}
\put(40,10) {\vector(-1,0){30}}
\put(10,100) {\vector(0,-1){90}}
\put(8,10) {\vector(0,1){90}}
\put(10,10){\vector(1,2){30}}
\put(10,100){\vector(1,-3){30}}
\put(15,100){$a$}
\put(45,70){$e$}
\put(45,40){$d$}
\put(45,10){$c$}
\put(10,0){$b$}
\put(-5,55){$N$}

\end{picture}

\caption{ Contraction of the arc $(a,b)$ in $N$ yields $N'$.}
\end{center}
\end{figure}
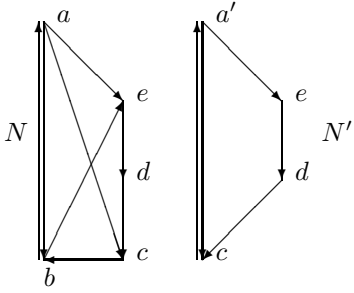

For example, in Figure 2, suppose $(a,b)$ in $N$ is contracted.  Both $(a,b)$ and $(b,a)$ are removed since there can be no loops.  There are new arcs $(a',e)$ because of $(a,e)$, $(a',c)$ because of $(a,c)$, and $(c,a')$ because of $(c,b)$.  There is no additional arc $(a',e)$ because of $(b,e)$.

Suppose $(u,v)$ is an arc of $N$ and $N'$ is obtained as above.  If neither $u$ nor $v$ is $a$ or $b$, then $(u,v)$ is an arc of $N'$, called the \emph{image} of $(u,v)$.  In case (ii) above we call $(u,a')$ the \emph{image} of $(u,a)$.  Similarly in cases (iii), (iv), and (v) the new arc is called the \emph{image} of the corresponding removed arc.   The arc $(a,b)$ has no image.

Suppose $B$ is a subset of $A$.  By \emph{successive contraction} of $B$ we mean:\\
(1) Select the arcs of $B$ in some order, say $(u_1,v_1), \cdots, (u_k,v_k)$.\\
(2) Recursively define networks $N_0, \cdots,  N_k$ and subsets $B_0, \cdots, B_k$ of $A(N_0), \cdots, A(N_k)$ as follows:\\
(a) $N_0 = N$, $B_0 = B$.\\
(b) Given $N_i$ with the collection of arcs $B_i$, let $N_{i+1}$ be the result of contracting the arc in $B_i$ which is the image of $(u_{i+1},v_{i+1})$ if such exists.   Moreover, $B_{i+1}$ is the image of $B_i$ in $N_{i+1}$.  (If there is no such arc, then $N_{i+1} = N_i$ and $B_{i+1} = B_i$.)

Thus $N_1$ is obtained by contracting $(u_1,v_1)$.   Then $N_2$ is obtained by contracting in $N_1$ the image of $(u_2,v_2)$, and so forth. 

\begin{thm} 
Suppose $N = (V,A,r,X)$ and $N' = (V',A',r', X)$ are $X$-networks and $f: N \to N'$ is a CSD map.  Then there is a collection $B \subseteq A$ of arcs of $N$ such that successive contraction of $B$ yields an $X$-network $N''$ which is $X$-isomorphic with $N'$.  
\end{thm}

\begin{proof} 
For each $v' \in V'$  let $W_{v'} =  f^{-1}(v')$.  
Then $W_{v'}$ is a connected  set of vertices in $N.$  Let $B_{v'}  = \{(a,b) \in A: a \in W_{v'}$ and $b \in W_{v'}\}$.  Thus $B_{v'}$ is the set of arcs of $N$ both of whose vertices lie in $W_{v'}$.  
Let $B  = \cup [ B_{v'}: v' \in V'] \subseteq A$.  

Let $\mathcal{P}$ denote the partition $\{W_{v'}\}$ of $V$, which is the kernel of $f$.  It is clear that for each $v'$,  successive contraction of the $B_{v'}$ contracts all vertices in $W_{v'}$ to a point.  Then $N'' = N/\mathcal{P}$.  By Theorem 3.5 there are CSD maps $g: N \to N/\mathcal{P}$ and $h: N/\mathcal{P} \to N'$ such that $f = h \circ g$.  The construction of Theorem 3.5 shows that $h$ is one-to-one and onto, and one verifies that $N/\mathcal{P} = N''$ is isomorphic with $N'$.
\end{proof} 

Conversely, if $N = (V,A,r,X)$ is an $X$-network and $B \subseteq A$, one may obtain a digraph $N'$ by successive contraction of $B$.  The network $N'$, however, does not need to be an $X$-network.  In order to obtain an $X$-network and a CSD map $f: N \to N'$, the arcs in $B$ must be chosen such that no two distinct leaves are ever identified, every member of $X$ remains a leaf,  and $N'$ remains rooted.

\section{Wired lifts} 

The next result, Theorem 4.1, shows that when $f: N \to N'$ is a CSD map, then in a certain sense the network $N'$ can ``almost" be identified as a subgraph in $N$.  In fact, there is a ``wired lift'' $M$ of $N'$ into $N$ consisting of an undirected subgraph $M$ of $N$ which resolves $N'$.   Typically there are numerous such wired lifts, at least one for any of a certain collection of arbitrary choices.  

More explicitly, let $G' = (V',E')$ be an (undirected) graph with leaf set $X$.  A graph $G = (V,E)$ with leaf set $X$ is a \emph{resolution} of $G'$ provided that $G'$ is obtained from $G$ by recursively contracting certain edges.  In each step, an edge $\{u,v\}$ of $G$ is contracted by removing the edge and identifying the two endpoints together.  No edge with an endpoint in $X$ is allowed to be contracted.  

Every graph is a resolution of itself.  

Let $N = (V,A,r,X)$ and $N' = (V',A',r',X)$ be $X$-networks.  Suppose $f: N \to N'$ is a surjective $X$-digraph map.  A \emph{wired lift} of $N'$ is an undirected subgraph $M =(W,E)$ of $Und(N)$ such that the following hold:\\
(1) For each arc $(u',v')$ of $N'$ there is exactly one arc $(u,v)$ of $N$ such the following three conditions hold: $f(u) = u'$, $f(v) = v'$, and $\{u,v\}$ is an edge of $M$.   The set of all edges $\{u,v\}$ so obtained will be denoted $E_1$ and the set of all vertices which occur in any of the edges $\{u,v\}\in E_1$ will be denoted $V_1'$.  Let $V_1 = V_1' \cup X$.\\
(2) Every edge $\{a,b\} \in E$ either lies in $E_1$ or else satisfies $f(a) = f(b)$.\\
(3)  For each vertex $v'$ of $N'$,  let $V(v') = \{w \in V_1: f(w) = v'\}$.  The induced subgraph $M[  f^{-1}(v') \cap W]$ is a tree with leafset contained in $V(v')$. 
 
We call $E_1$ the set of \emph{nondegenerate} edges of $M$, since the image under $f$ of each such edge is an edge of $N'$, not just a single vertex.  Note that $W \subseteq V$ and $E\subseteq E(Und(N))$.

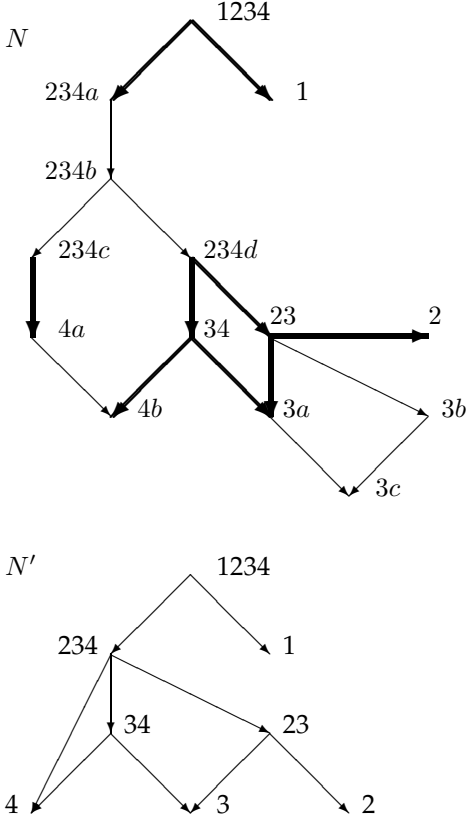
\begin{figure}[!htb]  
\begin{center}

\begin{picture}(170,320) (0,0)
\put(70,100){\vector(1,-1){30}}
\put(70,100){\vector(-1,-1){30}}
\put(40,70){\vector(0,-1){30}}
\put(40,70){\vector(2,-1){60}}
\put(40,70){\vector(-1,-2){30}}
\put(40,40){\vector(-1,-1){30}}
\put(40,40){\vector(1,-1){30}}
\put(100,40){\vector(-1,-1){30}}
\put(100,40){\vector(1,-1){30}}

\put(0,10) {4}
\put(80,10) {3}
\put(135,10) {2}
\put(45,40) {34}
\put(105,40) {23}
\put(20,70) {234}
\put(105,70) {1}
\put(80,100) {1234}
\put(0,100){$N'$}

\put(40,280){\vector(0,-1){30}}
\put(40,250){\vector(-1,-1){30}}
\put(40,250){\vector(1,-1){30}}
\put(10,190){\vector(1,-1){30}}
\put(100,190){\vector(2,-1){60}}
\put(100,160){\vector(1,-1){30}}
\put(160,160){\vector(-1,-1){30}}

\thicklines

\put(70,310){\vector(1,-1){30}}
\put(70,310){\vector(-1,-1){30}}
\put(10,220){\vector(0,-1){30}}
\put(70,220){\vector(0,-1){30}}
\put(70,220){\vector(1,-1){30}}
\put(70,190){\vector(-1,-1){30}}
\put(70,190){\vector(1,-1){30}}
\put(100,190){\vector(1,0){60}}

\put(71,310){\vector(1,-1){30}}    
\put(71,310){\vector(-1,-1){30}}
\put(11,220){\vector(0,-1){30}}
\put(71,220){\vector(0,-1){30}}
\put(71,220){\vector(1,-1){30}}
\put(71,190){\vector(-1,-1){30}}
\put(71,190){\vector(1,-1){30}}
\put(100,191){\vector(1,0){60}}
\put(101,190){\vector(0,-1){30}}
\put(100,190){\vector(0,-1){30}}

\thinlines
\put(80,310){1234}
\put(110,280){1}
\put(15,280){$234a$}
\put(15, 250){$234b$}
\put(20,220){$234c$}
\put(75,220){$234d$}
\put(20,190){$4a$}
\put(50,160){$4b$}
\put(75,190){$34$}
\put(100,195){$23$}
\put(160,195){$2$}
\put(105,160){$3a$}
\put(165,160){$3b$}
\put(140,130){$3c$}
\put(0,300){$N$}

\end{picture}

\caption{ There is a CSD map $f$ from $N$ to $N'$ given by  labeling the vertices of $N$ by vertices of $N'$. The bold arcs lie in $E_1$ and  correspond to arcs in $N'$.  A wired lift $M$ of $N'$ includes all of $Und(N)$ except the vertex $3b$ and the edges $\{23,3b\}$ and $\{3b,3c\}$.   }
\end{center}
\end{figure}

An example illustrating all these definitions is shown in Figure 3.  Figure 3 shows  a CSD map $f : N \to N'$ given by labeling the vertices of $N$ by vertices of $N'$ sometimes together with an additional letter.  For example, $f(234a)$ is the vertex of $N'$ labelled $234$.    One wired lift $M=(W,E)$ of $N'$ satisfies that $W$ consists of all vertices of $N$ except $3b$ while $E$ consists of all edges of $Und(N)$ except $\{23,3b\}$ and $\{3b,3c\}$.  Here $E_1$ corresponds to the nine edges of $Und(N)$ in bold, one for each arc of $N'$.  More precisely, $E_1$ contains $\{1234,1\}$, $\{1234,234a\}$, $\{234c,4a\}$, $\{234d,23\}$, $\{234d,34\}$, $\{23,2\}$, $\{23,3a\}$, $\{34,4b\}$ and $\{34,3a\}$.  Then $V_1'$ consists of all the vertices on any bold arc hence equals  $\{1, 1234, 234a, 234c, 4a, 234d, 34, 4b, 23, 3a, 2\}$.  Since the leaf $3c$ of $N$ is not in $V_1'$, we have $V_1 = V_1' \cup X$ $= V_1' \cup \{3c\}$.  Note that $f^{-1}(234)$ is a tree with vertex set $\{234a, 234b, 234c, 234d\}$ and $f^{-1}(4)$ is a tree with vertex set $\{4a, 4b\}$.  Both trees are included in $M$.  Finally $f^{-1}(3)$ is a tree with vertex set $\{3a, 3b, 3c\}$ but $M[f^{-1}(3) \cap W] = M[\{3a,3c\}]$ consists only of the edge $\{3a,3c\}$ and its vertices.

In (1) observe that if $x \in X$ and $(u',x)$ is an arc of $N'$, then there may be no arc $(u,x)$ of $N$ such that $f(u) = u'$.  The addition of $X$ to $V_1'$ may therefore be necessary.  In Figure 3 this occurred with the leaf $3\in X$, identified as vertex 3c of $N$.

Intuitively, $M$ is a subgraph of $Und(N)$ that is a resolution of $Und(N')$ in that for each vertex $v'$ of $N'$, $[f^{-1}(v')] \cap W$ consists of the vertices of a tree, all of whose vertices map to $v'$, not necessarily a single point.  The name ``lift'' suggests that $N'$ is being lifted into the domain of $f$.  Typically, not every vertex of $N$ will lie in $M$.

The following theorem gives sufficient conditions for a wired lift to exist given any choice of $E_1$.  The essential property is that $f$ be connected.  In order to have the possibility of always extending $E_1$ to a wired lift, the inverse image of each vertex of $N'$ must be connected.

\begin{thm}
Let $N = (V,A,r,X)$ and $N' = (V',A',r',X)$ be $X$-networks.  Suppose $f: N \to N'$ is a CSD map.  For each arc $(u',v')$ of $N'$ choose an arc $(u,v)$ of $N$ such that $f(u) = u'$, $f(v) = v'$. Let $E_1$ denote the set of edges $\{u,v\}$ of $Und(N)$ so obtained.  Then $f$ has a wired lift $M$ for which $E_1$ is the set of nondegenerate edges.  Each such wired lift $M$ is a resolution of $Und(N')$.  
\end{thm}

\begin{proof}
We may assume that $N'$ does not consist of a single vertex, so every vertex of $N'$ is an endpoint of some arc of $N'$.  Since $f$ is surjective, the construction of $E_1$ in the statement can be carried out.  Let $V_1$ be the set of all vertices of $N$ that arise as an endpoint of some edge in $E_1$ or else lie in $X$.   

For each vertex $v'$ of $N'$, the set $V(v') = \{w \in V_1: f(w) = v'\}$ is a subset of $f^{-1}(v')$.  Note that $V(v')$ is nonempty since each vertex occurs in some arc.  Since $f$ is connected, the graph $N_{v'} := Und(N)[f^{-1}(v')]$ is connected.   Consequently there exists a subtree $T_{v'}$ of $N_{v'}$ that contains $V(v')$, for example a minimal spanning tree.  We may assume that $T_{v'}$ has no leaves that are not members of $V(v')$ by removing other leaves.   Let $V_2$ denote the set of all vertices that lie on any $T_{v'}$, and let $E_2$ denote the set of all edges $\{u, v\}$ that lie in  $T_{v'}$ for some $v'$.  

Define the graph $M = (V_M, E_M)$  by $V_M := V_1 \cup V_2$ and $E_M := E_1 \cup E_2$.     

I claim $M$ is a wired lift.  Each edge $\{u,v\}$ in $E_2$ is contained in $V(v')$ for some $v'$ and satisfies $f(u) = f(v) = v'$.  Each edge $\{u,v\}$ in $E_1$ is such that either $(f(u),f(v))$ or $(f(v),f(u))$ is an arc of $N'$.  This shows that $M$ satisfies properties (1) and (2) of wired lifts.  Property (3) is immediate  since $T(v')$ is a tree.  

Finally, $M$ is a resolution of $Und(N')$ since, to obtain $Und(N')$ from $M$, one must merely contract every edge in $E_2$.  
\end{proof}

Observe that in the wired lift, the edges $E_1$ are in one-to-one correspondence with the edges of $Und(N')$.  All additional edges, {\it i.e.,} those in $E_2$, are such that both endpoints map under $f$ to the same vertex of $N'$.  Many different vertices of $M$ can project to the same vertex $v'$ in $N'$, but all those that do so lie on the tree $T_{v'}$.  

Even though $M$ is an undirected graph, each of the edges $\{u,v\}\in E_1$  may be considered to have a preferred orientation of either $(u,v)$ or $(v,u)$ depending on which is an arc of $N$.

\begin{figure}[!htb]  
\begin{center}

\begin{picture}(140,430) (0,0)
\put(100,420){\vector(1,-1){30}}
\put(100,420){\vector(-1,-1){30}}
\put(70,390){\vector(1,-1){30}}
\put(70,390){\vector(-1,-1){30}}
\put(40,360){\vector(-1,-1){30}}
\put(40,360){\vector(0,-1){30}}
\put(100,360){\vector(-1,-1){30}}
\put(100,360){\vector(1,-1){30}}
\put(100,360){\vector(1,0){30}}
\put(70,330){\vector(-1,-1){30}}
\put(70,330){\vector(0,-1){30}}
\put(130,330){\vector(1,-1){30}}
\put(130,330){\vector(0,-1){30}}
\put(70,300){\vector(-1,-1){30}}
\put(70,300){\vector(1,-1){30}}
\put(130,300){\vector(-1,-1){30}}
\put(130,300){\vector(1,-1){30}}
\put(100,270){\vector(0,-1){30}}
\put(100,240){\vector(1,-1){30}}
\put(100,240){\vector(-1,-1){30}}
\put(110,430){11}
\put(140,390){1}
\put(55,390){20}
\put(105,365){12}
\put(140,360){2}
\put(50,360){19}
\put(-1,330){10}
\put(50,330){9}
\put(80,330){18}
\put(140,330){13}
\put(30,300){8}
\put(80,300){17}
\put(115,300){14}
\put(170,300){3}
\put(110,270){15}
\put(170,270){4}
\put(50,270){7}
\put(110,240){16}
\put(60,210){6}
\put(140,210){5}
\put(30,390){$N$}

\put(100,190){\vector(1,-1){30}}
\put(100,190){\vector(-1,-1){30}}
\put(70,160){\vector(1,-1){30}}
\put(70,160){\vector(-1,-1){30}}
\put(40,130){\vector(-1,-1){30}}
\put(40,130){\vector(0,-1){30}}
\put(100,130){\vector(1,0){30}}
\put(100,130){\vector(0,-1){60}}
\put(100,70){\vector(-1,0){30}}
\put(100,70){\vector(-1,-1){30}}
\put(100,70){\vector(0,-1){30}}
\put(100,70){\vector(1,-1){30}}
\put(100,70){\vector(1,0){30}}
\put(100,40){\vector(-1,-1){30}}
\put(100,40){\vector(1,-1){30}}

\put(110,190){11}
\put(140,160){1}
\put(50,160){20}
\put(20,130){19}
\put(80,130){12}
\put(140,130){2}
\put(-1,100){10}
\put(50,100){9}
\put(60,70){8}
\put(105,80){[13]}
\put(140,70){3}
\put(60,40){7}
\put(140,40){4}
\put(105,40){16}
\put(60,10){6}
\put(140,10){5}

\put(25,160){$N'$}

\end{picture}

\caption{  Two $X$-networks $N$ and $N'$.  There is a CSD map from $N$ to $N'$.  
A wired lift $M$ consists of all edges of $Und(N)$ except $\{12,18\}$. }
\end{center}
\end{figure}
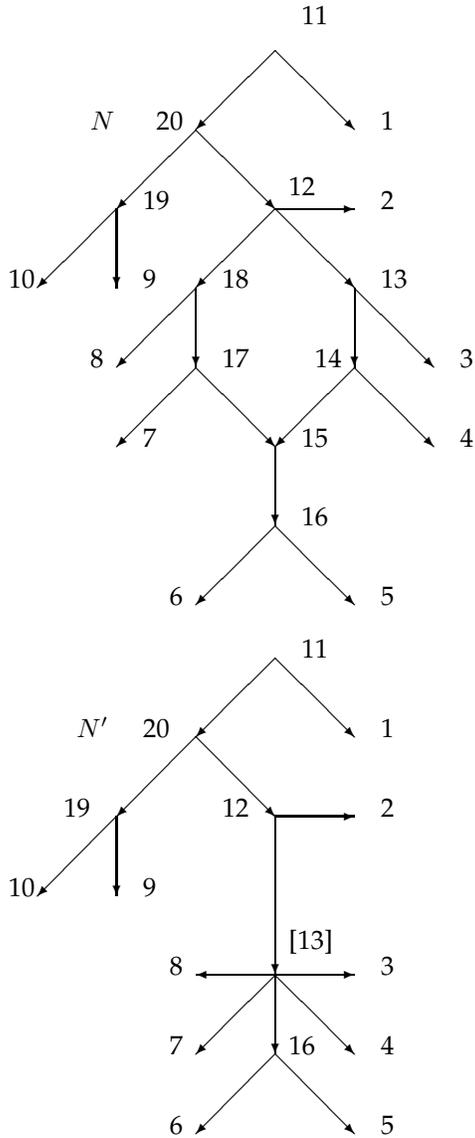

For another example, consider the networks $N$ and $N'$ in Figure 4
with $X = \{1, 2, 3, 4, 5, 6, 7, 8, 9, 10\}$.  
There is a CSD map $f: N \to N'$ given by
$f(x) = x$ for $x \in X$, $f(u) = u$ for $u \in \{11,12,16,19,20\}$, and $f(u) = [13]$ for $u \in \{13,14,15,17,18\}$.  A wired lift $M$ consists of all edges of $Und(N)$ except $\{12,18\}$.  Note that in $M$, 18 has no incoming directed arc from the directed graph $N$, but this is not a problem since the wired lift $M$ is an undirected graph.  
There is also a different wired lift, consisting of all edges of $Und(N)$ except $\{12,13\}$. 

The next few results show that a CSD map $\phi: N \to N'$ can put strong constraints on the structure of $N$.  

\begin{cor}
Let $N = (V,A,r,X)$ and $N' = (V',A',r',X)$ be $X$-networks and let $\phi: N \to N'$ be a CSD map.
Let $U'$ be an (undirected) subgraph of $Und(N')$ such that no vertex has total degree in $Und(U')$ greater than 3.  Then $Und(N)$ contains a subgraph $U$ homeomorphic with $U'$.
\end{cor}

\begin{proof}
Let $M$ be a wired lift of $N'$ into $N$. For each vertex $u'$ of $U'$, there are at most three edges of $U'$ with $u'$ as one endpoint.  If there are $k$ such edges, $k\leq 3$, then denote them $\{a_1',u'\}$, $\cdots$, $\{a_k',u'\}$.  Since $\phi$ is surjective, there are $k$ edges $\{a_i,u_i\}$  in $N$ for $i = 1, \cdots, k$, with $\phi(a_i) = a_i'$ and $\phi(u_i) = u'$.  Since $\phi^{-1}(u')$ is connected, there is a tree $T_{u'}$  in $\phi^{-1}(u')$ with vertex set containing $u_1, \cdots, u_k$ and with endpoints contained in $\{u_1, \cdots, u_k\}$.  Since $k\leq 3$, we may modify $T_{u'}$ if necessary so that no vertex has total degree in $T_{u'}$ greater than 3.  Thus $Und(N)$ contains a subgraph $U$ consisting of one edge for each edge of $U'$  together with a tree $T_{u'}$ for each vertex $u'$ of $U'$.  A simple consideration of cases shows that $U$ is homeomorphic with $U'$.  
\end{proof}

If $U'$ has a vertex $u'$ of total degree 4, then the corresponding tree $T_u'$  may  contain a vertex of total degree 4 but might instead contain only vertices of total degree 3, in which case there is no homeomorphism between $U$ and $U'$.  Effectively, it is possible that $U$ closely resembles $U'$ but resolves some vertices in $U'$ of total degree greater than 3.  

If $\{a,b,c,d\} \subseteq X$, the \emph{quartet} $ab|cd$ is the undirected tree with leaf set $\{a,b,c,d\}$ in which $a$ and $b$ share a neighbor and also $c$ and $d$ share a neighbor.  

\begin{cor}
Let $N = (V,A,r,X)$ and $N' = (V',A',r',X)$ be $X$-networks and let $\phi: N \to N'$ be a CSD map.
If $Und(N')$ contains a subgraph homeomorphic with the quartet $ab|cd$, then so does $Und(N)$.
\end{cor}

\begin{figure}[!htb]  
\begin{center}

\begin{picture}(140,200) (0,0)

\put(70,190) {\vector(1,-1){30}}
\put(70,190) {\vector(-1,-1){30}}
\put(40,160) {\vector(1,-1){30}}
\put(40,160) {\vector(-1,-1){30}}
\put(100,160) {\vector(1,-1){30}}
\put(100,160) {\vector(0,-1){30}}
\put(130,130) {\vector(-1,0){30}}
\put(100,130) {\vector(0,-1){30}}
\put(130,130) {\vector(0,-1){30}}
\put(75,190){5}
\put(30,160){6}
\put(1,130){1}
\put(105,160){7}
\put(90,130){8}
\put(135,130){9}
\put(90,100){3}
\put(55,130){2}
\put(135,100){4}
\put(10,190){$N$}

\put(70,70) {\vector(1,-1){30}}
\put(70,70) {\vector(-1,-1){30}}
\put(40,40) {\vector(0,-1){30}}
\put(40,40) {\vector(-1,-1){30}}
\put(100,40) {\vector(0,-1){30}}
\put(100,40) {\vector(1,-1){30}}
\put(75,70){5}
\put(30,40){6}
\put(105,40){789}
\put(1,10){1}
\put(45,10){2}
\put(90,10){3}
\put(135,10){4}
\put(10,70){$N'$}

\end{picture}

\caption{ There is a CSD map from $N$ to $N'$.  One wired lift of $N'$ is the subgraph $M$ of $Und(N)$ consisting of all edges except $\{9,8\}$.}
\end{center}
\end{figure}
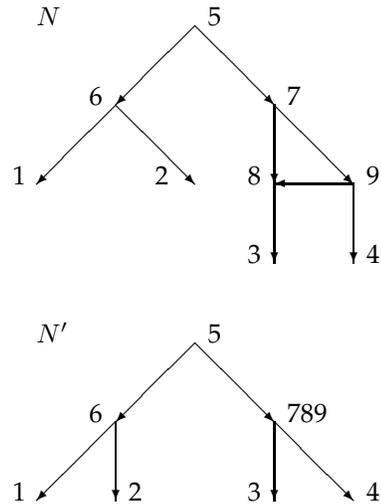

An illustration of Corollary 4.3 is given  in Figure 5.  The CSD map $f: N \to N'$ takes a vertex with label $a$ into a vertex whose label contains $a$.  $Und(N)$ contains a wired lift $M$ homeomorphic with the quartet $N'=12|34$.  One possibility for $M$ consists of all edges of $Und(N)$ except $\{9,8\}$.  The arc $(789,3)$ in $N'$ becomes the  path $7, 8, 3$ while the arc $(789,4)$ in $N'$ becomes the path $7, 9, 4$.  Note that $M$ is not an induced subgraph of $Und(N)$. Moreover, $M$ is only homeomorphic to $Und(N')$, not isomorphic to it.  Another possible wired lift includes all edges of $Und(N)$ except $\{7,8\}$. 

Consider the special case of a  CSD map from $N$ to a tree $T$.  Again, the structure of $T$ will be shown to put strong constraints on $N$.   Theorem 4.4 shows that if $N$ and $T$ are both binary  $X$-trees, then in fact $N$ and $T$ are the same tree.  

If $T$ is a rooted $X$-tree and $a$ and $b$  are in $X$, the \emph{most recent common ancestor} of $a$ and $b$, denoted mrca($a,b$), is the common ancestor of $a$ and $b$ such that no strict descendant is also a common ancestor of $a$ and $b$.  If $a$, $b$, $c$  are distinct members of $X$, we say that $T$ contains or displays the \emph{rooted triple} $ab|c$ provided that the most recent common ancestor of $a$ and $c$ is itself a strict ancestor of the most recent common ancestor of $a$ and $b$.   

\begin{thm}
Let $U$ and $T$ be rooted $X$-trees.  Suppose there is a CSD map $f: U \to T$. \\
(a) Every rooted triple $ab|c$ in $T$ is also a rooted triple of $U$.\\
(b)  If $T$ is binary, then $U$ and $T$ are homeomorphic.  
\end{thm}

\begin{proof}  The hypotheses mean that $T$ and $U$ are rooted $X$-trees in which there may be additional vertices with indegree 1 and outdegree 1 (which often are suppressed in trees).

We first show (a).  Without loss of generality we may assume that $12|3$ is in $T$.   We show $12|3$ in $U$ by considering other possibilities for $\{1,2,3\}$ in $U$.  

Suppose instead that $U$ displays $13|2$.  Let $a =$ mrca($1,2$) in $U$ and $b =$ mrca($1,3$) in $U$.  Let $c =$ mrca($1,3$) in $T$ and $d =$ mrca(1,2) in $T$.  Since $U$ displays $13|2$, in $U$ there is a directed path from $b$ to 1 and a directed path from $b$ to 3.  It follows that in $T$ there is a directed path from $f(b)$ to $f(1) = 1$ and from $f(b)$ to $f(3) = 3$.  Hence $f(b ) \leq$ mrca(1,3) $= c$ in $T$.  It follows that the image of the directed path in $U$ from $b$ to 1 is a directed path in $T$ from $f(b)$ to 1 which must pass through $d$.  In particular, $f^{-1}(d)$ must meet the path from $b$ to 1.  Similarly, in $U$ there is a directed path from $a$ to 2 and from $a$ to 3.  Hence $f(a) \leq$ mrca(2,3) $= c$ in $T$.  It follows that the directed path in $U$ from $a$ to 2 must be mapped into a directed path in $T$ from $f(a)$ to 2, which must pass through $d$.  Hence $f^{-1}(d)$ must meet the path from $a$ to 2.  

By hypothesis $f$ is connected, so $f^{-1}(d)$ is connected.   Since $f^{-1}(d)$ contains a point on the path from $a$ to 2 and also a point on the path from $b$ to 1 and $U$ is a tree, we see that $f(a) = f(b) = d$.  But this contradicts that $f(b) \leq c$.
This shows that $U$ cannot display $13|2$.  

A symmetric argument shows that $U$ cannot display $23|1$.  We wish to show $U$ displays $12|3$.  The remaining possibility is that $U$ displays the unresolved star 123.  In this case, let $a$ denote the star point in $U$.  In $U$ there is a directed path from $a$ to 1 and also from $a$ to 3.  Hence in $T$ there is a directed path from $f(a)$ to 1 and $f(a)$ to 3, so $f(a) \leq c$.  In particular the path from $a$ to 1 is taken to a path in $T$ that must pass through $d$, so the path from $a$ to 1 meets $f^{-1}(d)$.  Similarly in $U$ there is a directed path from $a$ to 2.  Its image in $T$ must pass through $d$, so the path from $a$ to 2 meets $f^{-1}(d)$.   Since $f^{-1}(d)$ is connected and $U$ is a tree, it follows that $f(a) = d$.  But this contradicts that $f(a) \leq c$.   Thus this possibility cannot arise.  This completes the proof of (a).

Part (b) follows from (a) since a rooted tree is determined up to homeomorphism by its rooted triples; see  \cite{bun71} or \cite{sem03}, p. 118.
\end{proof}

More generally, if $f: U \to T$ is a CSD map and both networks are $X$-trees, then $U$ possibly resolves some polytomies of $T$ but otherwise agrees with $T$.  The tree $U$ displayed in bold in Figure 1 together with the map $f|U: U \to N'$ shows that Theorem 4.4 is not true if $f$ is merely surjective but not connected since the binary trees $U$ and $N'$ are not homeomorphic.

If $N'$ is known and $f: N \to N'$ is a surjective digraph map but not connected, very little information about $N$ can be inferred.  Suppose $X$ is a finite set and $p: X \to  \mathbb{N}$ is a positive integral function.  The \emph{star network} with leaf set $X$ and \emph{multiplicity} $p(x)$  for $x \in X$ is the directed multigraph with vertex set $X \cup \{r\}$, root $r$ and $p(x)$ arcs $(r,x)$ for each $x \in X$; there are no other vertices or arcs.  The following theorem shows that any acyclic $X$-network $N'$ is the image of an $X$-network homeomorphic to a star network by a surjective digraph map.  Hence if $f: N \to N'$ is a surjective digraph map that is not connected, then  $N'$ puts negligible constraint on the structure of $N$.  

\begin{thm}  
Let $N' = (V',A',r',X)$ be an acyclic $X$-network.   There exists an $X$-network $N = (V,A,r,X)$ which is homeomorphic with a star network with leaf set $X$ such that there exists a surjective digraph map $f: N \to N'$.  
\end{thm}

\begin{proof}
For each $x \in X$, let $P(x)$ be the collection of directed paths in $N'$ from $r'$ to $x$.  Suppose there are $p(x)=|P(x)|$ such paths where, for $i = 1, \cdots, p(x)$ the $i$-th path has $k(x,i)$ arcs and is given by 
$r' =v_{(x,i,0)}$, $v_{(x,i,1)}$, $\cdots$, $v_{(x,i,k(x,i))} = x$.  Construct $N$ with $p(x)$ paths from $r$ to $x$, with no vertices in common except $r$ and $x$.   The $i$-th such path has vertices $r'$, $w_{(x,i,1)}$, $w_{(x,i,2)}$, $\cdots$, $w_{(x,i,k(x,i))} = x$.  Each arc of $N$ arises as an arc from such a path, and there are no other arcs.  There is a surjective digraph map $f: N \to N'$ given by $f(r) = r'$ and $f(w_{(x,i,j)}) = v_{(x,i,j)}$.  Note that $N$ is homeomorphic to a star network with $p(x)$ arcs from $r$ to $x$ and no other arcs.    
\end{proof}

See Figure 6 for an example.  In fact, instead of $P(x)$ one may use a subset of $P(x)$ such that each arc of $N'$ occurs in some path in some $P(x)$.

\begin{figure}[!htb]  
\begin{center}

\begin{picture}(150,210) (0,0)
\put(60,200){\vector(-2,-1){60}}
\put(60,200){\vector(-1,-1){30}}
\put(60,200){\vector(0,-1){30}}
\put(60,200){\vector(1,-1){30}}
\put(60,200){\vector(2,-1){60}}
\put(30,170){\vector(0,-1){30}}
\put(60,170){\vector(0,-1){30}}
\put(90,170){\vector(0,-1){30}}
\put(120,170){\vector(0,-1){30}}
\put(30,140){\vector(0,-1){30}}
\put(60,140){\vector(0,-1){30}}
\put(90,140){\vector(-1,-1){30}}
\put(120,140){\vector(0,-1){30}}
\put(0,160){1}
\put(20,110){2}
\put(20,140){$b$}
\put(20,170){$a$}
\put(50,170){$a$}
\put(50,140){$b$}
\put(50,110){$3$}
\put(100,170){$a$}
\put(100,140){$c$}
\put(130,170){$a$}
\put(130,140){$c$}
\put(130,110){$4$}
\put(60,210){$r$}
\put(30,200){$N$}

\put(100,100){\vector(-1,-1){30}}
\put(100,100){\vector(1,-1){30}}
\put(70,70){\vector(-1,-1){30}}
\put(70,70){\vector(1,-1){30}}
\put(40,40){\vector(-1,-1){30}}
\put(40,40){\vector(1,-1){30}}
\put(100,40){\vector(-1,-1){30}}
\put(100,40){\vector(1,-1){30}}
\put(80,95){$r'$}
\put(80,70){$a$}
\put(50,40){$b$}
\put(20,10){$2$}
\put(140,70){$1$}
\put(110,40){$c$}
\put(80,10){$3$}
\put(140,10){$4$}
\put(30,70){$N'$}

\end{picture}

\caption{  Two $X$-networks $N$ and $N'$.  There is a surjective digraph map $f$ from $N$ to $N'$ given by labeling each vertex $v$ of $N$ with the label of $f(v)$ in $N'$.  The map $f$ is not connected, and $N$ is homeomorphic to a star network.  None of the relationships in $N'$ between the leaves are present in $N$,  and there is no wired lift of $N'$ into $N$. }
\end{center}
\end{figure}
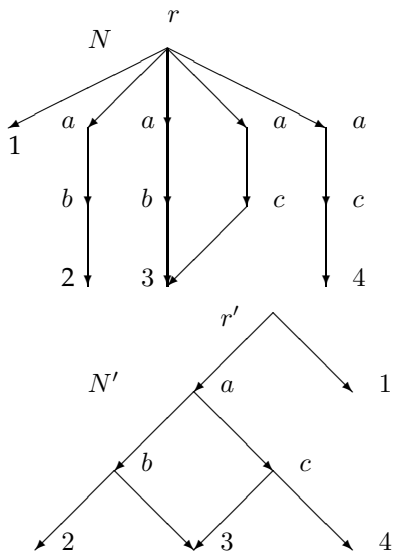

\section{ Successively Cluster-Distinct Networks}

Let $\mathcal{P}(X)$ denote the collection of subsets of $X$.  Following \cite{bar04}, given an $X$-network $N=(V,A,r,X)$, define the \emph{cluster map}  $cl: V \to \mathcal{P}(X)$ by
	$cl(v) = \{x \in X: v\leq x\}$, and call $cl(v)$ the \emph{cluster of v}.  
Sometimes for clarity $cl(v)$ will also be denoted $cl(v,N)$.  The taxon $v$ has the possibility of influencing the extant genomes for taxa in $cl(v)$ but  cannot influence the genomes of  taxa not in $cl(v)$.    

Call an $X$-network \emph{successively cluster-distinct} or more briefly \emph{cluster-distinct} if for each arc $(a,b)$ it is true that $cl(a) \neq cl(b)$. 
It is easy to construct an example of a successively cluster-distinct network $N$ with two vertices $a$ and $b$ for which $cl(a)=cl(b)$; the vertices $a$ and $b$ may not, however, be connected by an arc.  

Networks which are not successively cluster-distinct may have many successive vertices in a directed path all of which have the same cluster and hence potentially leave genetic influence on precisely the same extant vertices (members of $X$).  It may therefore be hard to distinguish their different genetic impacts on extant taxa.  Consequently it is plausible that we should simplify such a network in order to highlight features that are more likely distinguishable.  

The following algorithm Cluster-Distinct takes as input a network $N$ and essentially outputs a network $ClDis(N)$ which is successively cluster-distinct.  The idea is very simple.  Whenever $(u,v)$ is an arc and $cl(u,N) = cl(v,N)$, then $u$ and $v$ are identified.  Clearly $\{u,v\}$ is connected in $N$ since $(u,v)$ is an arc.  As a result of doing all such identifications, one obtains $ClDis(N)$.  

Here is a more precise description of the algorithm: 
\bigskip
\hrule height1pt
\medskip
\noindent\textbf{Algorithm} Cluster-Distinct\\
\textbf{Input}: $N = (V,A,r,X)$ is a network with leaf set $X$. \\
\textbf{Output}: A partition of $V$.  \\
\textbf{Procedure}: We construct a sequence $S_i$ of subsets of $V$. \\
(1) Let $S_0$ be the set of singleton sets from $V$.\\
(2)  Repeat recursively the following if any such step can be performed:
	Given $S_i$, suppose distinct $B_1$ and $B_2$ in $S_i$ satisfy that $u_1 \in B_1$, $u_2 \in B_2$ $(u_1,u_2)$ is an arc of $N$, and $cl(u_1,N) = cl(u_2,N)$.  Then $S_{i+1}$ is found by removing $B_1$ and $B_2$ from $S_i$ and adjoining $B_1 \cup B_2$.  Thus $S_{i+1} := (S_i -\{B_1, B_2\} ) \cup \{B_1 \cup B_2\}$. \\
(3)  Suppose for some $m$, $S_m$ has been constructed but there are no further ways to perform (2).  Return $S_m$. 
\medskip
\hrule height1pt
\bigskip

It is clear that $S_m$ is a partition of $V$.   

\begin{thm}
Let $N = (V,A,r,X)$ be an $X$-network.  Let $S_m$ denote the result of performing Algorithm Cluster-Distinct. \\
(1) $N/S_m$ is a successively cluster-distinct $X$-network.\\
(2) If $N$ is acyclic, then $N/S_m$ is acyclic. \\
(3) $S_m$ does not depend on the order in which the operations of Cluster-Distinct are performed. 
\end{thm}

\begin{IEEEproof}
(1) We first show that the partition $S_m$ is leaf-preserving.  For every $x\in X$, $cl(x) = \{x\}$.  Whenever a vertex $u$ is merged with a leaf $x$, $cl(u) = cl(x) = \{x\}$.  It follows that no two distinct leaves are equivalent.  Moreover, suppose $x\in X$ and $u\in [x]$ and $(u,v)$ is an arc.   Then $cl(u)=\{x\}$.  Since $(u,v)$ is an arc it follows that $cl(v)\subseteq cl(u)=\{x\}$ and since $cl(v)$ is nonempty we must have $cl(v)=\{x\}$. Hence $[v]=[u]=[x]$ so $v\in [x]$.  This proves that $S_m$ is leaf-preserving.    

By Theorem 3.1, $N/S_m$ is an $X$-network.  Note that if $u$ and $v$ are in $B \in S_m$, then $cl(u,N) = cl(v,N)$.  It is easy to see that $[u] \in N/S_m$ satisfies $cl([u],N/S_m)= cl(u,N)$.  To see that $N/S_m$ is successively cluster-distinct, suppose $([u],[v])$ is an arc of $N/S_m$.  Then there exist $u' \in [u]$ and $v' \in [v]$ 
with$(u',v')$ an arc of $N$.   If $cl(u',N) = cl(v',N)$ then by the algorithm $[u]$ and $[v]$ would be merged.  Hence $cl([u],N/S_m) \neq cl([v],N/S_m)$.  

(2) Suppose that there were a directed cycle $[u] = [u_0],$ $[u_1]$, $[u_2], \cdots,$ $[u_k]=[u]$ in $N/S_m$.  Then for $j = 0,  \cdots, k-1$, there exist $u_j'$ and $u_j''$ in $[u_j]$ such that $(u_j'',u_{j+1}')$ is an arc of $N$.  It is immediate that if $(w,v)$ is an arc in $N$, then $cl(w,N)$ contains $cl(v,N)$.  It follows that
$cl(u_0'',N)$ contains $cl(u_1',N) = cl(u_1'',N)$, which contains $cl(u_2',N) = cl(u_2'',N),$ $\cdots$,  which contains $cl(u_k',N) = cl(u_0'',N)$.  Hence all the clusters are the same whence algorithm Cluster-Distinct would merge them.   Thus $[u_0] = [u_1]= \cdots =[u_{k-1}]=[u_k]$.

(3) When the algorithm terminates, $S_m$ consists of the equivalence classes under the equivalence relation $\approx$ obtained as follows:\\
(a) First, define a relation $\sim$ on $V$ such that if $(u,v)$ is an arc, $v \notin X$, and $cl(u,N) = cl(v,N)$, then $u\sim v$ and $v\sim u$.\\
(b) $u \approx w$ iff either $u = w$ or else there exists a sequence $u_0, u_1, \cdots, u_k$ such that
$u = u_0$, $u_k = w$, and for $i = 0, \cdots, k-1$, $u_i \sim u_{i+1}$.\\
The equivalence classes clearly are independent of the order of operations.  Hence (3) follows.  
\end{IEEEproof}

Given $N$ we denote by $ClDis(N) := N/S_m$. Call $ClDis(N)$ the \emph{successively cluster-distinct network obtained from} $N$. 

\begin{cor}
There is a connected surjective $X$-digraph map $\phi: N \to ClDis(N)$.  Moreover, $ClDis(N)$ has a wired lift into $N$.
\end{cor}
\begin{proof}
 By induction, for all $i$, each member of $S_i$ is connected, whence each member of $S_m$ is connected. The result follows from Theorems 3.1 and  4.1.
\end{proof}

We call $\phi$ the \emph{natural projection} of $N$ onto $ClDis(N)$.

As a consequence, any wired lift of $\phi$ shows that $N$ has structure mimicking that of $ClDis(N)$.  Hence it may be natural to restrict attention in a given case to cluster-distinct networks.  Such networks are typically much simpler than the initial networks and exhibit much of the essential structure.

Algorithm Cluster-Distinct is very fast, indeed linear in the size of the network, as shown in the following theorem:

\begin{thm}  
 Let $N = (V,A,r,X)$ be an $X$-network.  Algorithm Cluster-Distinct may be carried out in time $O(|A|)$.  The cluster function $cl$ can be computed in time $O(|A|+|V|)$.  
\end{thm}

\begin{proof}
The function $cl$ may be computed as follows:   For each leaf $x \in X$, $cl(x) = \{x\}$.  Suppose $cl(c)$ is known for each child $c$ of $v$; then $cl(v) = \cup [cl(c): c $ is a child of $v]$.  Hence the cluster function may be computed in time $O(|A|+|V|)$.  

Once $cl$ is known, for each arc $(a,b)$ one identifies $a$ and $b$ precisely when $cl(a) = cl(b)$.  Hence the identifications may be carried out in time $O(|A|)$.  
\end{proof}

The next result describes an interesting property of $ClDis(N)$.

\begin{thm}
Let $N = (V,A,r,X)$ and $N' = (V',A',r',X)$ be $X$-networks.  Let $\phi:N \to ClDis(N)$ be the natural projection.  Suppose $f: N \to N'$ is a CSD map.  Assume whenever $(u,v) \in A$ 
and $cl(u,N) = cl(v,N)$ that it follows that $f(u) = f(v)$.  Then there exists a unique CSD map $g: ClDis(N) \to N'$ such that $f = g \circ \phi$.
\end{thm}
\begin{proof}
Let $\mathcal{P}$ and $\mathcal{Q}$ be respectively the kernels of $\phi$ and $f$.  By hypothesis, $\mathcal{P}$ is subordinate to $\mathcal{Q}$.  By Theorem 3.5 the desired map $g$ exists, and $g$ is a CSD map since both $f$ and $\phi$ are connected.  Uniqueness is immediate.  
\end{proof}

Let $N = (V,A,r,X)$ and $N' = (V',A',r',X)$ be $X$-networks. A CSD map $f: N \to N'$ with kernel $\mathcal{P}$ is \emph{cluster-distinct} if,  whenever $(u,v)$ is an arc of $N$ and $cl(u,N) = cl(v,N)$, 
then $f(u) = f(v)$.  A cluster-distinct map $f: N \to N'$ is \emph{universal (for cluster-distinct maps)} provided that given any cluster-distinct map $g: N \to N''$  there is a unique cluster-distinct map $h: N' \to N''$ such that $g = h \circ f$.

The essential content of Theorem 5.4 is that the natural projection map $\phi: N \to ClDis(N)$ is universal.  More explicitly, we have the following corollary:  

\begin{cor} 
Let  $N = (V,A,r,X)$ be an $X$-network. Let $\phi: N \to ClDis(N) $ be the natural projection map  Then $\phi$ is universal for cluster-distinct maps.
\end{cor}

\begin{proof}
Suppose $g: N \to N'$ is a cluster-distinct map.  
By Theorem 5.4, there exists a unique CSD map $h: ClDis(N) \to N'$ such that $g = h \circ \phi$.  Since $ClDis(N)$ is cluster-distinct, it is immediate that $h$ is cluster-distinct.  
\end{proof}

The network $ClDis(N)$ is in fact uniquely determined up to isomorphism by its universality property, as shown in the following result.  

\begin{thm} 
Let $N = (V,A,r,X)$ be an $X$-network.  Suppose $M$ is a phylogenetic $X$-network and $f:N \to M$ is a cluster-distinct CSD map which is universal.  Then $M$ is isomorphic with $ClDis(N)$.  
\end{thm}

\begin{proof} 
By Theorem 5.4 there is a unique CSD map $g: ClDis(N)\to M$ such that 
$f = g \circ \phi$.  Similarly, since $f$ is universal and $\phi:N \to ClDis(N)$ is a cluster-distinct map, there is a unique CSD map $h: M \to ClDis(N)$ such that $\phi = h \circ f$.  

It follows that
$ f = g \circ h \circ f$ and $\phi = h \circ g \circ \phi$.    Since $\phi = h \circ f$ we then have
$\phi = h \circ f = h \circ g \circ h \circ f$.
By uniqueness of the map $h$ it follows
$h \circ g \circ h = h$.

Similarly $f = g \circ \phi = g \circ h \circ f$.  By uniqueness of the map $g$ it follows
$g \circ h \circ g = g$.  

I claim that for all vertices $v$ of $ClDis(N)$ we have $(h\circ g)(v) = v$.
To see this, since $h$ is surjective there exists a vertex $u$ of $M$ such that $h(u) = v$.   But then $(h \circ g \circ h) (u) = h(u)$ so $(h \circ g) (v) = v$ for all $v$.

Similarly, for all vertices $u$ of $M$ we have $(g \circ h)(u) = u$.  To see this, since $g$ is surjective there exists a vertex $v$ of $ClDis(N)$ such that $g(v) = u$.  Hence
$(g \circ h \circ g )(v) = g(v)$ so $(g\circ h)(u) = u$ for all $u$.  

It follows that $g$ and $h$ are inverses of each other and hence isomorphisms of the networks . 
\end{proof}

For an example, consider the network $N$ in Figure 1.   Then $ClDis(N)$ is the tree $N'$  shown in Figure 1, and $f$ is the natural projection.  The image in $N'$ of each vertex in $N$ under the corresponding digraph map $\phi$ is indicated by the label of each vertex of $N$ in Figure 1.   In general, however, $ClDis(N)$ need not be a tree.

There are several interesting variants of Algorithm Cluster-Distinct.  One variant modifies step (2) so as never to identify a leaf with a parent having the same cluster.  Thus we replace (2) by ($2'$) as follows:\\
($2'$)  Repeat recursively the following if any such step can be performed:
	Given $S_i$, suppose distinct $B_1$ and $B_2$ in $S_i$ satisfy that $u_1 \in B_1$, $u_2 \in B_2$ $(u_1,u_2)$ is an arc of $N$,  $cl(u_1,N) = cl(u_2,N)$, and $u_2$ is not a leaf of $N$.  Then $S_{i+1} := (S_i -\{B_1, B_2\} ) \cup \{B_1 \cup B_2\}$. 

The advantage of $(2')$ is that tree-child leaves do not become hybrid in $ClDis(N)$.  

\section{Discussion}

This paper shows that the existence of a CSD map $f$ from $N$ to $N'$ implies interesting relationships between $N$ and $N'$.  By Theorems 3.4 and 3.5, CSD maps have good functorial properties; the composition of CSD maps is a CSD map, and certain CSD maps can be induced from other CSD maps.  
By Theorem 4.1, the CSD map implies the existence of a wired lift of $N'$ into $N$.  Such wired lifts show that some of the structure of $N'$ exists in $N$ as a ``skeleton''.  

Since Darwin, trees have been the primary method to describe phylogenies.  Now that hybridization and lateral gene transfer have been shown \cite{dag08}, \cite{doo07} to be important biologically, we need to consider other types of networks to be allowed in a useful analysis.  The true network $N$ containing each individual and all its progeny is the underlying reality, but such a network $N$ is too complicated to allow reconstruction from extant taxa.  A cartoon of  such a network $N$ is shown in Figure 1, in which $N'$ gives a plausible species tree for $N$.  In this case, $N' = ClDis(N)$.  

The author believes that, when one is trying to reconstruct a network $N$ from data, it is reasonable first to try to construct a successively cluster-distinct network $M$.  Theorem 5.1 shows that the cluster-distinct network $M=ClDis(N)$ always exists.    Often, as in the example of Figure 1, such a cluster-distinct network  will be much simpler than $N$ in having fewer vertices and arcs.  The additional hypothesis of cluster-distinctness can simplify the calculation of $M$.  If there is a CSD map from $N$ to $M$, then there will be a wired lift of $M$ in $N$, yielding properties of the more complicated network $N$ which may assist in its reconstruction.  

For a nontrivial example, in \cite{dre10} a cluster $C$ is called a \emph{tight} cluster of $N$ provided that $C$ is nonempty and whenever there is an undirected path from $c \in C$ to $d \in X-C$, then there exists a vertex $w$ on the path such that $cl(w) = C$.  It is easy to show that a cluster $C$ is a tight cluster of $N$ if and only if it is a tight cluster of $ClDis(N)$.

CSD maps exist whose images are trees.  Of special interest, however, is the possibility that there might be other classes of networks more general than trees but not as general as cluster-distinct networks.  For example, one might consider networks that are both cluster-distinct and tree-child \cite{crv09}.  Simple extensions of the results in this paper would lead to a CSD map from $N$ to such a network and a wired lift of such a network into $N$.  There are many other possibilities.

Future work should study more relationships between $N$ and $M$ if there is a CSD map from $N$ to $M$, possibly with additional assumptions.  

Other relationships between networks have been proposed, such as a reduction $R(N)$ of the network $N$ \cite{mor04}.  It is easy, however,  to construct examples showing that there need not be a CSD map from $N$ to $R(N)$.  

This paper explicitly dealt with networks with vertex set $V$ in which the set $X$ of species was in one-to-one correspondence with the set of leaves via a one-to-one map $\phi: X \to V$.  A more general notion of an $X$-network is that the map $\phi$ need not be one-to-one and must only have image containing the set of leaves.  In this situation most of the results go through with slightly different statements.  
A digraph map would require $f(\phi(x)) = \phi(x)$.

\textbf{Acknowledgments.}

I would like to thank Maria Axenovich for helpful references.  Thanks also to Francesc Rossell\'o and to the anonymous referees for helpful corrections and suggestions about an earlier version of this paper.

\begin{IEEEbiography}{Stephen Willson}
\includegraphics[width=1 in, bb = 0 0 295 329, clip]{WillsonSclose2.pdf}

Stephen J. Willson received his A.B. in Mathematics from Harvard in 1968.  In 1973 he received his Ph.D. in Mathematics from the University of Michigan in Ann Arbor.  His dissertation was in algebraic topology under the supervision of A.G. Wasserman.

He went to Iowa State University in Ames, Iowa in 1973, where he is currently University Professor and Janson Professor of Mathematics.  His research interests include phylogenetics, fractals, and game theory.  His hobbies include classical piano, choral singing, bird-watching, bicycling, and kayaking.  
\end{IEEEbiography}


\begin{thebibliography}{99}

\bibitem{ban92}
H.-J. Bandelt and A.  Dress,  (1992).  Split decomposition: a new and useful approach to phylogenetic analysis of distance data,  Molecular Phylogenetics and Evolution 1, 242-252.  

\bibitem{bar04}
M. Baroni,  C. Semple, and M. Steel, (2004),  A framework for representing reticulate evolution,  Annals of Combinatorics 8, 391-408.  

\bibitem{bun71}
P. Buneman, (1971), The recovery of trees from measures of dissimilarity.  In: Mathematics in the Archaeological and Historical Sciences (ed. F.R. Hodson, D.G. Kendall, and P. Tautu), Edinburgh University Press, Edinburgh, pp. 387-395.    

\bibitem{clr08}
G. Cardona, M. Llabr\'es, F. Rossell\'o, and G. Valiente, (2008), A distance metric for a class of tree-sibling phylogenetic networks, Bioinformatics 24, 1481-1488.  


\bibitem{crv09}
G. Cardona, F.  Rossell\'o,  and G. Valiente, (2009),  Comparison of tree-child phylogenetic networks, 
IEEE/ACM Transactions on Computational Biology and Bioinformatics, 6(4): 552-569.  


\bibitem{dag08}
T. Dagan, Y. Artzy-Randrup, and W. Martin, (2008), Modular networks and cumulative impact of lateral transfer in prokaryote genome evolution, Proc. Natl. Acad. Sci. USA. 105, 10039-10044.  

\bibitem{dan08}
A. Daneshgar, H. Hajiabolhassan, and N. Hamedazimi, (2008), On connected colourings of graphs, Ars Combinatoria 89, 115-126. 


\bibitem{doo07}
W. F. Doolittle and E. Bapteste, (2007), Pattern pluralism and the Tree of Life hypothesis, Proc. Natl. Acad. Sci. USA. 104, 2043-2049.  


\bibitem{dre10}
A. Dress, V. Moulton, M. Steel, and T. Wu, (2010), Species, clusters and the `tree of life': a graph-theoretic perspective, Journal of Theoretical Biology 265(4): 535-542.   


\bibitem{gel04a} 
D. Gusfield, S. Eddhu, and C. Langley, (2004),  Optimal, efficient reconstruction of phylogenetic networks with constrained recombination,  Journal of Bioinformatics and Computational Biology 2, 173-213.  

\bibitem{hah97}
G. Hahn and C. Tardif, (1997). Graph homomorphisms: structure and symmetry, in Graph Symmetry: Algebraic Methods and Applications (G. Hahn and G. Sabidussi, eds) 
NATO Adv. Sci. Inst. Ser. C: Math. Phys. Sci., vol. 497, 
Kluwer Academic Publishers, Dordrecht, 1997, pp. 107-166.  

\bibitem{hel04}
P. Hell and J. Ne\v set\v ril, (2004),  Graphs and Homomorphisms,  Oxford University Press, Oxford.  

\bibitem{ier09}
L. J. J. van Iersel, J. C. M. Keijsper, S. M. Kelk, L. Stougie, F. Hagen, and T. Boekhout, (2009),  Constructing level-2 phylogenetic networks from triplets,  IEEE/ACM Transactions on Computational Biology and Bioinformatics, 6(43): 667-681.  


\bibitem{mor04}
B.M.E. Moret,  L.  Nakhleh, T. Warnow, C.R.  Linder, A. Tholse, A. Padolina, J.  Sun, and R. Timme, (2004),  Phylogenetic networks: modeling, reconstructibility, and accuracy,  IEEE/ACM Transactions on Computational Biology and Bioinformatics 1, 13-23.  

\bibitem{mor09}
D.A. Morrison, (2009), Phylogenetic networks in systematic biology (and elsewhere).  In R.M. Mohan (ed.) Research Advances in Systematic Biology (Global Research Network, Trivandrum, India) pp. 1-48.  


\bibitem{nak04}
L. Nakhleh, T.  Warnow, and C.R. Linder, (2004), Reconstructing reticulate evolution in species--theory and practice, in  P.E. Bourne and D. Gusfield, eds., Proceedings of the Eighth Annual International Conference on Computational Molecular Biology (RECOMB '04, March 27-31, 2004, San Diego, California), ACM, New York, 337-346.  

\bibitem{sem03}
C. Semple and M. Steel, (2003),  Phylogenetics,  Oxford University Press, Oxford.  

\bibitem{wan01}
L. Wang, K.  Zhang, and L. Zhang, (2001),  Perfect phylogenetic networks with recombination,   Journal of Computational Biology 8,  69-78. 



\bibitem{wil12}
S.J. Willson (2012), CSD Homomorphisms Between Phylogenetic Networks,   IEEE/ACM Transactions on Computational Biology and Bioinformatics  9, 1128-1138.



\end{thebibliography}
\end{document}